\RequirePackage{fix-cm}

\documentclass[20pt,a4paper]{svjour3}       
\usepackage[table, svgnames, dvipsnames]{xcolor}
\usepackage{makecell, cellspace, caption}
\usepackage{tabularx, ragged2e, booktabs, caption}
\usepackage{amssymb}
\usepackage{amsmath}
\usepackage{color,soul}

\usepackage{geometry}
\makeatletter
\renewcommand{\subsection}{%
  \@startsection{subsection}
    {2}
    {\z@}
    {-21dd plus-8pt minus-4pt}
    {10.5dd}
    {\normalsize\bfseries\boldmath}%
}
\makeatother

\smartqed  
\usepackage{tabularx,ragged2e,booktabs,caption}
\newcolumntype{C}[1]{>{\Centering}m{#1}}

  \usepackage{hyperref}
\usepackage{graphicx}

 \journalname{}
\usepackage{pdflscape}
\begin{document}
	
\title{Construction of $(\sigma,\delta)$-cyclic codes over a non-chain ring and their applications in DNA codes}
	
	
\author{Ashutosh Singh, Priyanka Sharma and Om Prakash$^*$}
\authorrunning{Singh et al.}

\institute{ \at $^1$
              Department of Mathematics\\
              Indian Institute of Technology Patna, Patna 801 106, India \\
              \email{ashutosh$\_$1921ma05@iitp.ac.in, priyanka$\_$2121ma13@iitp.ac.in, om@iitp.ac.in(*corresponding author)}
}
\date{Received: date / Accepted: date}

	\maketitle
	
	\begin{abstract}
	For a prime $p$ and a positive integer $m$, let $\mathbb{F}_{p^m}$ be the finite field of characteristic $p$, and  $\mathfrak{R}_l:=\mathbb{F}_{p^m}[v]/\langle v^l-v\rangle$ be a non-chain ring. In this paper, we study the $(\sigma,\delta)$-cyclic codes over $\mathfrak{R}_l$. Further, we study the application of these codes in finding DNA codes. Towards this, we first define a Gray map to find classical codes over $\mathbb{F}_{p^m}$ using codes over the ring $\mathfrak{R}_l$. Later, we find the conditions for a code to be reversible and a DNA code using $(\sigma, \delta)$-cyclic code. Finally, this algebraic method provides many classical and DNA codes of better parameters.
	\end{abstract}

\maketitle

\section{Introduction}

DNA computing is an evolving area that uses the principles of molecular biology to carry out computations. After the work of Adleman \cite{Adleman}, many researchers started working in this direction. The fundamental aspect of these studies is to perform parallel computations and also store large amounts of data efficiently for a longer time. These studies have vast applications in computing, data storage, the development of new medical technologies, cryptography, and several others. However, the computations in DNA molecules are also prone to errors and they can arise due to various reasons, such as errors in DNA synthesis, errors in DNA hybridization, and errors in DNA sequencing \cite{MKGupta}. To mitigate these errors, various techniques have been developed, such as error-correcting codes, self-assembly, and error-tolerant algorithms. In the last two decades, the construction of error-correcting codes using the algebraic method has attracted researchers. \\
Algebraic coding theory \cite{Mac} is a mathematical discipline that pertains to the design and analysis of error-correcting codes utilized to ensure the integrity of transmitted data through communication channels that are susceptible to noise and interference. The theory uses algebraic structures, including groups, rings, and fields, to represent codes as subsets of larger algebraic structures, such as finite vector spaces \cite{CodeRing,ore}. Recently, following the work of Hammon et al. \cite{HammonsR}, many authors considered commutative and noncommutative rings for the construction of codes. In \cite{mostafanasab,non-chain siap,abualrub,bayram,bennenni,guenda}, authors have studied cyclic codes over some rings and further extended it to find a condition for DNA codes. Furthermore, in the quest to find good codes, researchers have studied skew cyclic codes for classical and DNA codes over different finite rings \cite{Boucher07,D09,leroy,gursoy,Ashu}. These studies also helped in the construction of self-dual and quantum codes \cite{Verma}. Recently, Patel and Prakash \cite{shikha2} studied codes over rings with $\theta$-derivation.
However, these rings have some limitations as the length of the Gray image has only certain combinations, resulting in a limited scope of good codes. Hence, this work extends the study to construct the codes over more general rings. These codes have flexible lengths and produce optimal or better codes compared to the codes over other commutative and noncommutative rings \cite{Grassl,shikha2,shikha}. We further extend the study to the application of finding DNA codes and several MDS, optimal and better classical codes using Sagemath \cite{sage} and Magma \cite{magma} computational software.

This paper is organized as follows: Section 2 discusses the structure of the ring $\mathfrak{R}_l$ and presents some basic definitions and results. Section 3 provides a necessary and sufficient condition for $(\sigma,\delta)$-cyclic codes over the ring $\mathfrak{R}_l$. In Section 4, we present results on reversibility and identify the condition for $(\theta,\delta)$-cyclic codes to be reversible, while Section 5 focuses on the results of DNA codes obtained from $(\theta,\delta)$-cyclic codes. In Section 6, we offer several examples of classical codes and DNA codes to support our study, along with new and improved classical codes and DNA codes listed in Tables [\ref{Tbl1}, \ref{Tbl2}, \ref{Tbl3}, \ref{DNA codes}]. Section 7 concludes the work.

\section{Preliminary}\label{sec2}

Let $\mathbb{F}_{p^m}$ be the field of characteristic $p$ where $p$ is a prime and $m$ is a positive integer. Consider the ring $\mathfrak{R}_l=\mathbb{F}_{p^m}[v]/\langle v^l-v \rangle$, where $l \geq 2$. Thus $\mathfrak{R}_l$ is a non-chain ring, and its elements are of the form $a_0+va_1+\cdots+v^{l-1}a_{l-1}$ where $a_i$'s are in $\mathbb{F}_{p^m}$.

Given a primitive element $\omega$ of the field $\mathbb{F}_{p^m}$ and $\varsigma=\omega^{\frac{p^m-1}{l-1}}$ where $(l-1)|(p^m-1)$, consider
	\begin{align*} \gamma_1&=1-v^{l-1},\\
	\gamma_2&=\frac{1}{l-1}(v+v^2+\cdots+v^{l-2}+v^{l-1}),\\
	\gamma_3&=\frac{1}{l-1}(\varsigma v+\varsigma^2 v^2+\cdots+\varsigma^{l-2} v^{l-2}+v^{l-1}),\\
	\gamma_4&=\frac{1}{l-1}(\varsigma^2 v+ (\varsigma^2) ^2 v^2+\cdots+ (\varsigma^2)^{l-2} v^{l-2}+v^{l-1}),\\
	&\vdots \\
	\gamma_l&=\frac{1}{l-1}(\varsigma^{l-2} v+ (\varsigma^{l-2})^2 v^2+\cdots+ (\varsigma^{l-2})^{l-2} v^{l-2}+v^{l-1}). \end{align*}
 Then,  we get
 $$\gamma_i \gamma_j=\begin{cases}
\gamma_i, & \text{if $i=j$}\\
0, & \text{if $i \neq j$}
\end{cases} \text{ and } \sum_{i=1}^l \gamma_i=1 \text{ for all $1\le i,j\le l$};
$$
and the set $\{\gamma_i:1\leq i\leq l\}$ contains pairwise orthogonal idempotent elements.

 Therefore, by Chinese remainder theorem, we can decompose the ring $\mathfrak{R}_l$ as follows $$\mathfrak{R}_l=\bigoplus_{i=1}^{l}\gamma_i\mathfrak{R}_l\cong \bigoplus_{i=1}^{l}\gamma_i\mathbb{F}_{p^m}.$$
Clearly, $\gamma_i\mathfrak{R}_l\cong \gamma_i \mathbb{F}_{p^m}$ and any element $r$ of the ring $\mathfrak{R}_l$ can be uniquely written as $r=\sum_{i=1}^l\gamma_ir_i$ where $r_i\in \mathbb{F}_{p^m}$ for all $i$.
\begin{definition}
    A linear code $\mathcal{C}$ over a ring $\mathcal{R}$ of length $n$ is defined as an $\mathcal{R}$-submodule of $\mathcal{R}^n$ and elements of $\mathcal{C}$ are called codewords.
\end{definition}
The Hamming distance of two codewords $x$ and $y$ is defined as the minimum number of substitutions required to make $y$ from $x$. The Hamming distance of a code is defined as the minimum of the distances between two distinct codewords.
\begin{definition}
    Let $\theta$ be an automorphism of a finite ring $\mathcal{R}$. Then a map $\delta:\mathcal{R}\to \mathcal{R}$ is called a $\theta$-derivation if the following conditions hold:
    \begin{enumerate}
        \item $\delta(r_1+r_2)=\delta(r_1)+\delta(r_2)$ and
\item $\delta(r_1r_2)=\delta(r_1)r_2+\theta(r_1)\delta(r_2)~~~~~~ \text{ for all $r_1,r_2\in \mathcal{R}$}.$
    \end{enumerate}
\end{definition}

Consider the set of all polynomials over ring $\mathfrak{R}_l$ with an automorphism $\sigma$ and an inner derivation $\delta$. Then this set forms a noncommutative ring under the usual addition of polynomials, and multiplication is defined with respect to $xr=\sigma(r)x+\delta(r)$. This ring is known as a skew polynomial ring and is denoted by $\mathfrak{R}_l[x;\sigma,\delta]$. As $\mathbb{F}_{p^m}$ is a division ring, $\mathbb{F}_{p^m}[x;\theta, \delta]$ forms a right (resp. left) Euclidean ring. Throughout this paper, we take $\theta$ as the Frobenius automorphism of $\mathbb{F}_{p^m}$, and we extend this automorphism to find an automorphism $\sigma$ of the ring $\mathfrak{R}_l$. It is easy to verify that the order of $\theta$ and $\sigma$ is $2$, i.e., $\theta^2=Id_{\mathbb{F}_{p^m}}$ and $\sigma^2=Id_{\mathfrak{R}_l}$.\\
Now, we give a few results related to a derivation map which are helpful throughout the paper.
\begin{lemma}
    Let $\delta:\mathcal{R}\to \mathcal{R}$ be a map defined by $\delta(r)=\alpha(\theta(r)-r)$
where $\alpha$ is a non-zero element in $\mathcal{R}$. Then the map $\delta$ is a $\theta$-derivation of $\mathcal{R}$.
\end{lemma}
\begin{proof}
 Let $r_1,r_2\in \mathcal{R}$. Then for a nonzero $\alpha$ in $\mathcal{R}$, we get
\begin{align*}
    \delta(r_1+r_2)&=\alpha(\theta(r_1+r_2)-(r_1+r_2))\\
    &=\alpha(\theta(r_1)-r_1)+\alpha(\theta(r_2)-r_2)\\
    &=\delta(r_1)+\delta(r_2)
\end{align*}
and
\begin{align*}
    \delta(r_1r_2)&=\alpha(\theta(r_1r_2)-r_1r_2)\\
    &=\alpha(\theta(r_1r_2)-\theta(r_1)r_2+ \theta(r_1)r_2-r_1r_2)\\
    &=\alpha(\theta(r_1)\theta(r_2)-\theta(r_1)r_2)+\alpha(\theta(r_1)r_2-r_1r_2)\\
    &=\alpha\theta(r_1)(\theta(r_2)-r_2)+\alpha(\theta(r_1)-r_1)r_2\\
    &=\delta(r_1)r_2+\theta(r_1)\delta(r_2).
\end{align*}
\end{proof}

\begin{remark}
    For the prime field $\mathbb{F}_p$ of the field $\mathbb{F}_{p^m}$, any derivation is only zero derivation. Interestingly, with inner derivation and the Frobenius automorphism over $\mathbb{F}_4$, we get four possible noncommutative rings.
\end{remark}
\begin{lemma}
    Let $\delta$ be the inner derivation map defined above and $\theta$ fixes $\alpha$. Then for any $r\in \mathcal{R}$, we get $\delta(\theta(r))=\theta(\delta(r))$.
\end{lemma}
\begin{proof} Let $r=a+vb$. Then
\begin{align*}
    \delta(\theta(r))&=\delta(a^p+vb^p)\\
    &=\alpha(a^{p^2}+vb^{p^2}-a^p-vb^p),
\end{align*}
and
\begin{align*}
    \theta(\delta(r))&=\theta(\alpha(\theta(r)-r))\\
    &=\theta(\alpha(a^p+vb^p-a-vb))\\
    &=\theta(\alpha)\theta(a^p+vb^p-a-vb)\\
    &=\alpha(a^{p^2}+vb^{p^2}-a^p-vb^p).
\end{align*}
Therefore, $\delta(\theta(r))=\theta(\delta(r))$ holds when $\theta(\alpha)=\alpha$.
\end{proof}
 	




\begin{lemma} Let $\delta$ be a  $\theta$-derivation of a ring $\mathcal{R}$, which commutes with $\theta$. Then
\begin{align*}
      x^na = ~&  {}^nC_0 \theta^n(a) x^n + {}^nC_1(\theta^{n-1}\delta)(a)x^{n-1} +{}^nC_2(\theta^{n-2}\delta^{2})(a)x^{n-2}  + {}^nC_3(\theta^{n-3}\delta^{3})(a)x^{n-3} \\ & + \cdots  + {}^nC_{n-1}(\theta\delta^{n-1})(a)x + {}^nC_n\delta^n(a),
\end{align*}
where multiplication is defined as the composition of mappings.
\end{lemma}
\begin{proof}
    The proof follows from induction on $n.$
\end{proof}

\begin{corollary}
    Consider skew polynomial ring $\mathbb{F}_4[x;\theta,\delta]$. Then for any $a\in \mathbb{F}_4$, we have
    $$ x^na= \begin{cases}
        (\theta(a)x+\delta(a))x^{n-1},&\textit{if n is odd},\\
        ax^n,& \textit{if n is even.}    \end{cases} $$
\end{corollary}
\begin{proof}
 Using definitions of $\theta$ and $\delta$ and solving inductively for $n$, the proof is straight-forward.\end{proof}
Next, we define a Gray map $\Phi: \mathfrak{R}_l \longrightarrow \mathbb{F}_{p^m}^l$ as

\begin{equation} \label{gray}
	\Phi(a_0+a_1v+\cdots+a_{l-1}v^{l-1})=(a_0,a_1,\ldots,a_{l-1})N,
\end{equation}
where $a_i \in \mathbb{F}_{p^m}~ \forall i\in\{0,1,\ldots,l-1\}$ and $N\in GL_l(\mathbb{F}_{p^m})$ such that $NN^T=\beta I_l, ~\beta \in \mathbb{F}_{p^m}^*$. Clearly, the function $\Phi$ is an $\mathbb{F}_{p^m}$-linear distance preserving map, and this map is extendable component-wise to $\mathfrak{R}^n_l$.

 \renewcommand{\arraystretch}{1.7}

\begin{table}[ht]

\begin{center}\caption{$(\theta,\delta)$-cyclic codes of length $n$ over $\mathbb{F}_{p^m}$}\label{Tbl1}

\begin{tabular}{|c | c |c|c|c|c|c|}
	\hline
$(n,p^m)$ & $\delta(a), a\in \mathbb{F}_{p^m}$ &$g(x)$ & $\mathcal{C} $&Remarks\\
  \hline
$(16,9)$ & $t^2(\theta(a)-a)$ & $1t^3t^7t1$ & $[16,12,4]_{9}$ & Almost MDS\\
  	\hline
  $(12,25)$ & $t(\theta(a)-a)$ & $t^4t^21$ & $[12,10,3]_{25}$ & MDS\\
  	\hline
   $(12,25)$ & $t(\theta(a)-a)$ & $t^{16}t^{17}t^{23}2t^{22}1$ & $[12,7,5]_{25}$ & Almost MDS\\
  	\hline
   $(20,25)$ & $t(\theta(a)-a)$ & $t^8t1$ & $[20,18,3]_{25}$ & MDS\\
  	\hline
   $(24,25)$ & $t^2(\theta(a)-a)$ & $t^{11}t^31$ & $[24,22,3]_{25}$ & MDS\\
  	\hline
   $(14,49)$ & $t^2(\theta(a)-a)$ & $t^{43}t^91$ & $[14,12,3]_{49}$ & MDS\\
  	\hline
   $(16,49)$ & $t(\theta(a)-a)$ & $t^5t^{15}1$ & $[16,14,3]_{49}$ & MDS\\
  	\hline
   $(21,49)$ & $t^2(\theta(a)-a)$ & $t^{20}t^{19}1$ & $[21,19,3]_{49}$ & MDS\\
  	\hline

\end{tabular}

\label{overfield}
\end{center}

\end{table}

\section{$\boldsymbol{(\sigma,\delta)}$-cyclic codes over $\mathfrak{R}_l$}
In this section, first we discuss the structure of $(\theta,\delta)$-cyclic codes over $\mathbb{F}_{p^m}$. Further we derive the structure of $(\sigma,\delta)$-cyclic codes over the ring $\mathfrak{R}_l$.\\
\begin{definition}
    A $(\theta,\delta)$-linear code of length $n$ is said to be a $(\theta,\delta)$-cyclic code if for any codeword $c=(c_0,c_1,\ldots,c_{n-1})\in \mathcal{C}$, we have $(\theta,\delta)$-cyclic shift $\tau_{(\theta,\delta)}(c)=(\theta(c_{n-1})+\delta(c_0),\theta(c_{0})+\delta(c_1),\ldots,\theta(c_{n-2})+\delta(c_{n-1}))\in \mathcal{C}$.
\end{definition}
\begin{theorem}
    A code $\mathcal{C}$ of length $n$ over $\mathbb{F}_{p^m}$ is a $(\theta,\delta)$-cyclic code if and only if code $\mathcal{C}$ is a left $\mathbb{F}_{p^m}[x;\theta,\delta]$-submodule of $\frac{\mathbb{F}_{p^m}[x;\theta,\delta]}{\langle x^n-1\rangle}$.
\end{theorem}
\begin{proof}
Since the code $\mathcal{C}$ is a linear code, for any $c(x),c'(x) \in \mathcal{C}$ we get $c(x)+c'(x)\in \mathcal{C}$. For a codeword $c(x)=c_0+c_1x+\cdots+c_{n-1}x^{n-1}$ in the $(\theta,\delta)$-cyclic code $\mathcal{C}$, $xc(x)=\sum_{i=0}^{n-1}[\theta(c_{i-1})+\delta({c_i})]x^i \in \mathcal{C}$ where $-1\equiv n-1.$ Therefore, for a polynomial $a(x)\in \mathbb{F}_{p^m}[x;\theta,\delta]$ and $c(x)\in \mathcal{C}$, $a(x)c(x) \in \mathcal{C}$ and the code $\mathcal{C}$ is a left $\mathbb{F}_{p^m}[x;\theta,\delta]$-submodule. \\
Conversely, suppose the code $\mathcal{C}$ is a left $\mathbb{F}_{p^m}[x;\theta,\delta]$-submodule of $\frac{\mathbb{F}_{p^m}[x;\theta,\delta]}{\langle x^n-1\rangle}$ then by the definition of submodule it is linear. Now for any $c(x)\in \mathcal{C}$, take $a(x)=x\in \mathbb{F}_{p^m}[x;\theta,\delta]$, then $xc(x)\in \mathcal{C}$. Therefore, the code $\mathcal{C}$ is a $(\theta,\delta)$-cyclic code.
\end{proof}

\begin{theorem}
    Let $\mathcal{C}$ be a left submodule of $\mathbb{F}_{p^m}[x;\theta,\delta]/\langle x^n-1\rangle$. Then $\mathcal{C}$ is generated by a monic polynomial $g(x)$ of minimum degree, where $g(x)$ right divides $x^n-1$ over $\mathbb{F}_{p^m}[x;\theta,\delta]$.
\end{theorem}
\begin{proof}

 Let $g(x)$ be a monic polynomial of minimum degree in $\mathcal{C}$. Then for a polynomial $c(x)\in \mathcal{C}$, there exist unique polynomials $q(x)$ and $r(x)$ in $\mathbb{F}_{p^m}[x;\theta,\delta]$ such that $$c(x)=q(x)g(x)+r(x)$$
where $\deg(r(x))<\deg(g(x))$ or $r(x)=0$. Since polynomials $c(x)$ and $g(x)$ belong to $\mathbb{F}_{p^m}[x;\theta,\delta]$-submodule $\mathcal{C}$, polynomial $r(x)=c(x)-q(x)g(x)\in \mathcal{C}$ and we get $r(x)=0$. Therefore, the code $\mathcal{C}$ is generated by $g(x)$.\\
Now consider polynomial $x^n-1$ in place of $c(x)$, then we get $x^n-1=q(x)g(x)$ implies $g(x)|_r (x^n-1)$, that is, polynomial $g(x)$ right divides $x^n-1$.
\end{proof}
Codes given in the above theorem are free modules of dimension $k=n-\deg(g(x))$. And the rows of the generator matrix of these codes are given by $\tau_{\theta,\delta}^i(c)$ where $0\leq i\leq k-1$.\\

 \begin{example}
Let $\mathbb{F}_{7^2}=\mathbb{F}_{7}(t)$ be the field of order $49$. We define $\theta \in Aut(\mathbb{F}_{49})$ as $\theta(a)=a^{7}$ and let $\delta(a) =t^2(\theta(a)-a)$ where $a\in \mathbb{F}_{49}$.\\
Let $\mathcal{C}$ be a $(\theta, \delta)$-cyclic code of length $21$ over $\mathbb{F}_{49}$ generated by $g(x) = x^2 + t^{19}x + t^{20}$, where the polynomial $g(x)$ right divides $x^{21}-1$ in $\mathbb{F}_{49}[x;\theta,\delta]$. The factorization of $x^{21}-1$ in terms of $g(x)$ is given by
\begin{align*}
    x^{21}-1 = &(x^{19}+t^{13}x^{18}+t^7x^{17}+t^{10}x^{16}+t^2x^{15}+t^{45}x^{14}+t^{29}x^{13}+t^7x^{12}+t^2x^{11}+t^{44}x^{10}
+t^{36}x^9\\ &+t^{17}x^8+t^{32}x^7+t^{46}x^{6}+t^{44}x^5+t^{30}x^4+t^{24}x^3+t^{41}x^{2}+t^{26}x+t^{37})(x^2 + t^{19}x + t^{20})
\end{align*}
 Then the code $\mathcal{C}$ is an MDS $(\theta, \delta)$-cyclic code with the parameters $[21,19,3]_{49}$. Note that, a code with the parameters $[n,k,d]$ is a maximum distance separable (MDS) code if the condition $n+1=k+d$ is satisfied. And if the length of the code is $k+d$ then the code is almost MDS.
\end{example}

Now, in the following theorem, we give the structure of $(\sigma,\delta)$-cyclic codes over the ring $\mathfrak{R}_l$. Here $\sigma$ is the extension of the automorphism $\theta$ of $\mathbb{F}_{p^m}$ to the ring $\mathfrak{R}_l$.

\begin{theorem}
    Let $\mathcal{C}=\sum_{i=1}^l\gamma_i\mathcal{C}_i$ be a linear code over $\mathfrak{R}_l$ where, for $i\in\{1,\ldots, l\}$, $\mathcal{C}_i$'s are linear codes over the field $\mathbb{F}_{p^m}$. Then the code $\mathcal{C}$ is $(\sigma,\delta)$-cyclic code over $\mathfrak{R}_l$ if and only if $\mathcal{C}_i$'s are $(\theta,\delta)$-cyclic codes over $\mathbb{F}_{p^m}$.
\end{theorem}
\begin{proof}
Let $\mathcal{C}=\sum_{i=1}^l\gamma_i\mathcal{C}_i$ be a $(\sigma,\delta)$-cyclic codes of length $n$ over the ring $\mathfrak{R}_l$. Let $a=\sum_{i=1}^l\gamma_ic_i$, where $c_i=(c_{i0},c_{i1},c_{i2},\ldots, c_{i(n-1)}) \in \mathcal{C}_i$ and $a=(a_0,a_1,a_2,\ldots,a_{n-1})\in \mathcal{C}$. Then for all $i\in \{1,\ldots,l\}$ we have $$\sigma(a)=\sigma\left(\sum_{i=1}^l\gamma_ic_i\right)=\sum_{i=1}^l\gamma_i\theta(c_i)
$$
and
\begin{align*}
    \delta\left(\sum_{i=1}^l\gamma_ic_i\right)&=\alpha\left(\sigma\left(\sum_{i=1}^l\gamma_ic_i\right)-\sum_{i=1}^l\gamma_ic_i\right)\\
    &=\alpha\left(\sum_{i=1}^l\gamma_i\theta(c_i)-\sum_{i=1}^l\gamma_ic_i\right)\\
    &= \alpha\left(\sum_{i=1}^l\gamma_i(\theta(c_i)-c_i)\right)\\
    &=\sum_{i=1}^l\gamma_i\delta(c_i).
\end{align*}
Now we get
\begin{align*}
    \tau_{\sigma,\delta}(a)=&\tau_{\sigma,\delta}\left(\sum_{i=1}^l\gamma_ic_i\right)\\
    =&\Bigg( \sigma\left(\sum_{i=1}^l\gamma_ic_{i(n-1)}\right)+\delta\left(\sum_{i=1}^l\gamma_ic_{i0}\right),\sigma\left(\sum_{i=1}^l\gamma_ic_{i0}\right)+\delta\left(\sum_{i=1}^l\gamma_ic_{i1}\right)\\& ,\ldots,
    \sigma\left(\sum_{i=1}^l\gamma_ic_{i(n-2)}\right)+\delta\left(\sum_{i=1}^l\gamma_ic_{i(n-1)}\right)\Bigg)\\
    =& \sum_{i=1}^l\gamma_i\tau_{\theta,\delta}( c_i).
\end{align*}
Therefore, the codes $\mathcal{C}_i$'s are $(\theta,\delta)$-cyclic codes.\\
Conversely, suppose the codes $C_i$'s are $(\theta,\delta)$-cyclic codes. Then for any $a=(a_0,a_1,\ldots,a_{n-1})\in \mathcal{C}$, we get $a_j=\sum_{i=1}^l\gamma_ic_{ij}$ where $j\in \{0,1,\ldots,n-1\}$
and
\begin{align*}
    \tau_{\sigma,\delta}(a)=& (\sigma(c_{n-1})+\delta(c_0),\sigma(c_0)+\delta(c_1),\ldots,
    \sigma(c_{n-2})+\delta(c_{n-1}))\\
    =&\Bigg(\sigma\left(\sum_{i=1}^l\gamma_ic_{i(n-1)}\right)+\delta\left(\sum_{i=1}^l\gamma_ic_{i0}\right),\sigma\left(\sum_{i=1}^l\gamma_ic_{i0}\right)+\delta\left(\sum_{i=1}^l\gamma_ic_{i1}\right)\\&,\ldots,
    \sigma\left(\sum_{i=1}^l\gamma_ic_{i(n-2)}\right)+\delta\left(\sum_{i=1}^l\gamma_ic_{i(n-1)}\right)\Bigg)\\
    =&\bigg(\sum_{i=1}^l\gamma_i\theta(c_{i(n-1)})+\sum_{i=1}^l\gamma_i\delta(c_{i0}),\ldots,
    \sum_{i=1}^l\gamma_i\theta(c_{i(n-2)})+\sum_{i=1}^l\gamma_i\delta(c_{i(n-1)})\bigg)\\
    =& \sum_{i=1}^l\gamma_i\tau_{\theta,\delta}(c_i).
\end{align*}
Since $\mathcal{C}_i$'s are $(\theta,\delta)$-cyclic codes, $\tau_{\theta,\delta}(c_i)\in \mathcal{C}_i$. Therefore, code $\mathcal{C}$ is a $(\sigma,\delta)$-cyclic code.

\end{proof}

\begin{theorem}
    Let $\mathcal{C}=\sum_{i=1}^l\gamma_i\mathcal{C}_i$ be a $(\sigma,\delta)$-cyclic code of length $n$ over $\mathfrak{R}_l$ where $\mathcal{C}_i$'s are linear codes over the field $\mathbb{F}_{p^m}$. Then the code $\mathcal{C}$ is principally generated by $f(x)=\sum_{i=1}^l\gamma_if_i(x),$ where  $f_i(x)$ right divides $x^n-1$ in $\mathbb{F}_{p^m}[x;\theta,\delta]$. In particular, $f(x)$ right divides $x^n-1$ in $\mathfrak{R}_l[x;\sigma,\delta].$
\end{theorem}
\begin{proof}
    Since $C_i$'s are $(\theta,\delta)$-cyclic code of length $n$ over the field $\mathbb{F}_{p^m}$, then for some polynomial $f_i(x)$ with $f_i(x)|_r(x^n-1)$ we get $C_i=\langle f_i(x)\rangle$ where $i\in \{1,\ldots,l\}$. Further take $f(x)=\sum_{i=1}^l\gamma_if_i(x)$. Then $\langle \gamma_i f_i(x)\rangle\subseteq \mathcal{C}$ implies $\langle f(x)\rangle\subseteq \mathcal{C}$. Again from
    $$\gamma_if(x)=\gamma_if_i(x), \hfil i\in \{1,\ldots, l\},$$
    we get $\mathcal{C}\subseteq \langle f(x)\rangle$. Therefore, the code $\mathcal{C}$ is principally generated by the polynomial $ f(x)$.\\
    Similarly, for right divisors $f_i(x)$ of $x^n-1$ in $\mathbb{F}_{p^m}[x;\theta,\delta]$ we get $f(x)=\sum_{i=1}^l\gamma_if_i(x)$ as a right divisor of $x^n-1$ in $\mathfrak{R}_l[x;\sigma,\delta]$.
\end{proof}

\renewcommand{\arraystretch}{1.9}

\begin{table}[ht]
\begin{center}\caption{$(\sigma,\delta)$-cyclic codes of length $n$ over $\mathfrak{R}_2$ and their Gray images}\label{Tbl2}
\begin{tabular}{|c| c |c|c|c|c|c|}
	\hline
$(n,q)$ & $\alpha(\theta(a)-a)$ &$g_1(x),~ g_2(x)$&  $\Phi(\mathcal{C})$&Remarks\\
  \hline
  $(6,9)$ & $t^2(\theta(a)-a)$ & $t^51$, $t^6t^21$ & $[12,9,3]_9$ & Optimal\\
  	\hline
  $(6,9)$ & $t^2(\theta(a)-a)$ & $t^51$, $t^3tt^51$ & $[12,8,4]_9$ & Optimal\\
	\hline
	 $(10,25)$ & $t^2(\theta(a)-a)$ & $t^{16}1$, $t^4t^{11}1$ & $[20,17,3]_{25}$ & Almost MDS\\
  	\hline
    $(10,25)$ & $t^2(\theta(a)-a)$ & $t^{20}t^{13}1$, $tt^{10}1$ & $[20,16,4]_{25}$ & Almost MDS\\
  	\hline
   $(12,25)$ & $t(\theta(a)-a)$ & $t^21$, $t^{22}t^31$ & $[24,21,3]_{25}$ & Almost MDS\\
  	\hline
   $(24,25)$ & $t(\theta(a)-a)$ & $t^21$, $t11$ & $[48,45,3]_{25}$ & Almost MDS\\
  	\hline
   $(6,49)$ & $t(\theta(a)-a)$ & $t^21$, $t^91$ & $[12,10,3]_{49}$ & MDS\\
  	\hline
  $(6,49)$ & $t(\theta(a)-a)$ & $t^21$, $t^9t^{19}1$ & $[12,9,4]_{49}$ & MDS\\
	\hline

 $(21,49)$ & $t(\theta(a)-a)$ & $61$, $2t41$ & $[42,38,4]_{49}$ & Almost MDS\\
	\hline

\end{tabular}
\label{2copies}
\end{center}
\end{table}

 \renewcommand{\arraystretch}{1.9}

\begin{table}
\begin{center}\caption{Skew cyclic
codes of length $n$ over $\mathbb{F}_{p^m}$ and $\mathfrak{R}_l$ and their Gray images}
\label{Tbl3}
\begin{tabular}{|c|c|c|c|c|c|c|}
	\hline
$(n,p^m)$ & $\alpha(\theta(a)-a)$ &$g_1(x),g_2(x),\ldots, g_l(x)$&  $\Phi(\mathcal{C})$&Comparison\\
  \hline
  $(21,4)$ & $t(\theta(a)-a)$ & $100011,~11$ & $[42,36,4]$ & Optimal \\

  \hline
$(28,4)$ & $0(\theta(a)-a)$ & $10t^211,~t^21$ & $[56,50,4]$ & Optimal \\
	\hline
$(40,9)$ & $t(\theta(a)-a)$ & $t^32101,~1t^31$ & $[80,74,4]$ & Optimal \\
  \hline

 $(24,16)$ & $t^2(\theta(a)-a)$ & $t^{11}1$ & $[24,23,2]$ & $[24,22,2]$\cite{shikha}\\
	\hline
 $(12,16)$ & $t(\theta(a)-a)$ & $t0t^7t^{13}1,~ t^{14}t^{11}t^{10}1$ & $[24,17,6]$ & $[24,16,6]$\cite{shikha}\\
	\hline
 $(40,16)$ & $t^5(\theta(a)-a)$ & $t^{12}1$ & $[40,39,2]$ & $[40,37,2]$\cite{shikha}\\
	\hline
 $(20,16)$ & $t^3(\theta(a)-a)$ & $t^21,~ t^{12}t1$ & $[40,37,3]$ & $[40,35,3]$\cite{shikha}\\
	\hline
 $(20,16)$ & $t^3(\theta(a)-a)$ & $t^9t^71,~ t^2t^{13}11$ & $[40,35,4]$ & $[40,34,4]$\cite{shikha}\\
	\hline

\end{tabular}
\end{center}
\end{table}

 \begin{remark}
In the Tables [\ref{overfield}, \ref{2copies}, \ref{Tbl3}], we write a polynomial by its coefficients.  The coefficients of a polynomial are framed in the ascending order of the variable's exponents. In particular, we substitute $c_0c_1c_2$ in place of the polynomial $c_0+c_1x+c_2x^2$.
\end{remark}
 \section{Reversible and DNA codes over $\mathbb{F}_4$} \label{rev}
 In this section, we delve into the theory of reversible codes over a skew polynomial ring with a derivation. Reversible codes have several advantages over other types of error-correcting codes. This class of codes has a simple decoding algorithm that can be implemented efficiently. It is also robust to errors that occur in bursts, making them useful in situations where errors are likely to occur in clusters.

\begin{definition}
	Let $g(x)=g_0+g_1 x+\cdots+g_mx^m$ be a polynomial in $\mathbb{F}_4[x;\theta,\delta]$. Then $g(x)$ is said to be a palindromic polynomial if $g_i=g_{m-i}$ and
a $(\theta,\delta)$-palindromic if $g_i=\theta(g_{m-i})-\delta(g_{m-i+1})$ where $ i \in \{1,2,\ldots,m\}$.
\end{definition}

Now, using the above definitions, we derive results for the reversible codes in $\mathbb{F}_4[x;\theta,\delta]$.

\begin{theorem}
    Let $\mathcal{C}$ be a $(\theta, \delta)$-cyclic code of even length $n$ generated by a polynomial $g(x)=g_0+g_1 x+\cdots+g_mx^m$ over $\mathbb{F}_4$, where $m$ is odd. Then $\mathcal{C}$ will be reversible if $g(x)$ is palindromic and $\delta(g((x)) \in \mathcal{C}$.
\end{theorem}
\begin{proof}
    Let $\mathcal{C} = \langle g(x) \rangle$ be a $(\theta, \delta)$-cyclic code of even length $n$, where $g(x) =\sum_{i=0}^{m}g_ix^i,$ $m$ is odd. Let $c \in \mathcal{C}$ be a codeword. Then $c = \sum_{j=0}^{k-1}a_jx^jg(x)$, where $a_j$'s $\in \mathbb{F}_4$.
 \begin{align*}
     \noindent c =& \sum_{j=0}^{k-1} a_j x^j \sum_{i=0}^{m}g_i x^i\\
        =& \sum_{j=0}^{k-1} a_j \sum_{i=0}^{m}\theta^j(g_i) x^{i+j}+ \sum_{j=0}^{(k-3)/2} a_{2j+1} \sum_{i=0}^{m}\delta(g_i) x^{i+2j}.
        \end{align*}
       Now, we find the reverse of above codeword $c$, which is given by
     \begin{align*}
     c^r = & \sum_{j=0}^{k-1} a_j \sum_{i=0}^{m}\theta^j(g_i) x^{n-1-(i+j)}+ \sum_{j=0}^{(k-3)/2} a_{2j+1}\sum_{i=0}^{m}\delta(g_i) x^{n-1-(i+2j)}\\
      = & \sum_{j=0}^{(k-1)/2} a_{2j} \left[\sum_{i=0}^{m}g_i x^{m-i}\right]x^{k-1-2j}+\sum_{j=0}^{(k-3)/2} a_{2j+1}x \left[\sum_{i=0}^{m}g_i x^{m-i}\right]x^{k-3-2j} \\&+ \sum_{j=0}^{(k-3)/2} a_{2j+1} \Bigg[\Big(\sum_{i=0}^{m}\delta(g_i) x^{m-i}\Big)x^{k-1-2j}+\left(\sum_{i=0}^{m}\delta(g_i) x^{m-i}\right)x^{k-3-2j} \Bigg].
    \end{align*}
    Further, using the palindromic condition of the polynomial $g(x)$, we get
    \begin{align*}
        c^r=& \sum_{j=0}^{(k-1)/2} a_{2j} g(x)x^{k-1-2j}+\sum_{j=0}^{(k-3)/2} a_{2j+1}x g(x)x^{k-3-2j}\\
      & + \sum_{j=0}^{(k-3)/2} a_{2j+1} \left[\delta(g(x)) x^{k-1-2j}+\delta(g(x)) x^{k-3-2j} \right]\\
      =& \sum_{j=0}^{k-1} a_{j}x^{k-1-j} g(x) + \sum_{j=0}^{(k-3)/2} a_{2j+1} \delta(g(x))\big[ x^{k-1-2j}+ x^{k-3-2j} \big]\\
      =& \sum_{j=0}^{k-1} a_{j}x^{k-1-j} g(x) + \sum_{j=0}^{(k-3)/2} a_{2j+1} \big[ x^{k-1-2j}+ x^{k-3-2j} \big] \delta(g(x)).
    \end{align*}
    If $\delta(g(x))\in \mathcal{C}$, then $c^r\in \mathcal{C}$. Thus, the code $\mathcal{C}$ is reversible.
\end{proof}

\begin{theorem}
    Let $\mathcal{C}$ be a $(\theta, \delta)$-cyclic code of even length $n$ generated by a monic polynomial $g(x)=\sum_{i=0}^{m}g_i x^i$ over $\mathbb{F}_4$, where $m$ is even. Then $\mathcal{C}$ will be reversible if $g(x)$ is $(\theta,\delta)$-palindromic and $\Big( \sum_{i=0}^{m} g_i x^{m-i} \Big)x \in \mathcal{C}$.
\end{theorem}
\begin{proof}
 Let $\mathcal{C} = \langle g(x) \rangle$ be a $(\theta, \delta)$-cyclic code of even length $n$, where $g(x)=\sum_{i=0}^{m}g_i x^i$ such that $m$ is even. Let $c \in \mathcal{C}$ be a codeword. Then $c = \sum_{j=0}^{k-1}a_jx^jg(x)$, where $a_j$'s $\in \mathbb{F}_4$.
    \begin{align*}
        c = & \sum_{j=0}^{k-1} a_j x^j \sum_{i=0}^{m}g_i x^i\\
       = & \sum_{j=0}^{k-1} a_j \sum_{i=0}^{m}\theta^j(g_i) x^{i+j}+ \sum_{j=0}^{(k-2)/2} a_{2j+1} \sum_{i=0}^{m}\delta(g_i) x^{i+2j}
       \end{align*}
Now, the reverse of the codeword will be
       \begin{align*}
       c^r = & \sum_{j=0}^{k-1} a_j \sum_{i=0}^{m}\theta^j(g_i) x^{n-1-(i+j)}+ \sum_{j=0}^{(k-2)/2} a_{2j+1}
      \sum_{i=0}^{m}\delta(g_i) x^{n-1-(i+2j)}\\
     = & \sum_{j=0}^{(k-2)/2} a_{2j} \left[\sum_{i=0}^{m}g_i x^{m-i}\right]x^{k-1-2j}
    + \sum_{j=0}^{(k-2)/2} a_{2j+1} \Bigg[\bigg(\sum_{i=0}^{m-1}\Big(\theta(g_i)+\delta(g_{i+1})\Big)\\
      & x^{m-i} +\theta (g_m)\Big)x^{k-2-2j}\Bigg].
    \end{align*}
    Now, using the $(\theta, \delta)$-palindromic condition of the polynomial $g(x)$, we get
    \begin{align*}
      c^r  = & \sum_{j=0}^{(k-2)/2} a_{2j} \left[\sum_{i=0}^{m}g_i x^{m-i}\right]x^{k-1-2j}
     + \sum_{j=0}^{(k-2)/2} a_{2j+1} \Bigg[\bigg(\sum_{i=0}^{m}g_{m-i} x^{m-i}\Big)x^{k-2-2j}\Bigg]\\
      = & \sum_{j=0}^{(k-2)/2} a_{2j} x^{k-2-2j}\left[\sum_{i=0}^{m}g_i x^{m-i}\right]x
     + \sum_{j=0}^{(k-2)/2} a_{2j+1} x^{k-2-2j} \Bigg[\bigg(\sum_{i=0}^{m}g_{m} x^{m}\bigg)\Bigg]\\
     = & \sum_{j=0}^{(k-2)/2} a_{2j} x^{k-2-2j}\left[\sum_{i=0}^{m}g_i x^{m-i}\right]x
     + \sum_{j=0}^{(k-2)/2} a_{2j+1} x^{k-2-2j} g(x).
    \end{align*}
    If $\Big( \sum_{i=0}^{m} g_i x^{m-i} \Big)x \in \mathcal{C}$, then  $c^r\in \mathcal{C}$. Thus, the code $\mathcal{C}$ is reversible.
\end{proof}

In this part, we discuss the complementary condition for reversible codes. A code is a DNA code if it satisfies both reversible and complement conditions. In this work, we identify the DNA nucleotides by the field elements using an identification: $1\to A,~ t^2 \to T,~ 0 \to G$ and $t \to C$.\\
Now, in the following lemma, we give a relation between the alphabets and their complement using the above identification.
	

\begin{lemma} For given $(\theta,\delta)$-cyclic code in Section \ref{rev}, the following conditions hold:
	\begin{itemize}
		\item[(1)] For any $r\in \mathbb{F}_4,\ r^c=r+t.$
        \item[(2)] For any $r_1, r_2\in \mathbb{F}_4,\ r_1^c+r_2^c=(r_1+r_2)^c+t$.
	\end{itemize}
\end{lemma}
\begin{proof} This lemma can easily be proved by observing the correspondence of DNA nucleotides.
\end{proof}

\begin{remark}
We identify $\frak{i}_n(x)$ by the polynomial $1+x+x^2+\dots+x^{n}$.
\end{remark}


\begin{theorem} Given a polynomial $a(x)$ of degree $n$ in $\mathbb{F}_4[x]$. Then $$a(x)^{rc}=a(x)^r+ t \frak{i}_{n}(x).$$
\end{theorem}
\begin{proof} Let $a(x)=a_0+a_1x+\cdots + a_{n-1}x^{n-1}+x^n$ be a polynomial of degree $n$ in $\mathbb{F}_4[x]$ where $a_0$ is a non-zero element of $\mathbb{F}_4$. Then
\begin{align*}
a(x)^{rc}=&a_n^c+a_{n-1}^cx+\cdots+a_1^cx^{n-1}+a_0^cx^{n}\\
=&(a_n+t)+(a_{n-1}+t)x+(a_{n-2}+t)x^2+\cdots +(a_1+t)x^{n-1}+(a_0+t)x^n\\
=&a(x)^r+t \frak{i}_n(x).
\end{align*}

\end{proof}



The following corollary is obvious from the above theorems.
\begin{corollary}
Let $\mathcal{C}$ be a cyclic code of length $m$ over $\mathbb{F}_4$. If the code $\mathcal{C}$ is reversible and all-$t$ vector (that is, the value of each entry is $t$) is in $\mathcal{C}$, then $\mathcal{C}$ is a DNA code.
\end{corollary}

\begin{table}
\begin{center}\caption{ DNA codes from skew cyclic codes of length $n$ over $\mathbb{F}_{p^m}$}
\label{DNA codes}
\begin{tabular}{|c|c|c|c|}
	\hline
$(n,p^m)$ & $g(x)$&  $\Phi(\mathcal{C})$\\
  \hline
$(12,4)$ & $x^5 + t^2x^3 + t^2x^2 + 1$ & $[12, 7, 4]$  \\
	\hline
$(18,4)$ & $x^{11} + tx^9 + t^2x^8 + x^7 + x^4 + t^2x^3 + tx^2 + 1$ & $[18, 7, 8]$  \\
	\hline
 $(20,4)$ & $x^7 + t^2x^6 + tx^5 + tx^2 + t^2x + 1$ & $[20, 13, 4]$  \\
	\hline
  $(20,4)$ & $x^9 + t^2x^8 + x^7 + tx^6 + t^2x^5 + t^2x^4 + tx^3 + x^2 + t^2x + 1$ & $[20, 11, 6]$  \\
	\hline
  $(20,4)$ & $x^{11} + tx^{10} + x^9 + t^2x^7 + x^6 + x^5 + t^2x^4 + x^2 + tx + 1$ & $[20, 9, 8]*$  \\
	\hline
 $(22,4)$ & $x^{11} + x^8 + x^7 + tx^6 + tx^5 + x^4 + x^3 + 1$ & $[22, 11, 8]*$  \\
	\hline
  $(24,4)$ & $x^7 + t^2x^6 + t^2x^5 + tx^4 + tx^3 + t^2x^2 + t^2x + 1$ & $[24, 17, 4]$  \\
	\hline
   $(24,4)$ & $x^9 + t^2x^8 + x^7 + tx^5 + tx^4 + x^2 + t^2x + 1$ & $[24, 15, 6]$  \\
	\hline
   $(28,4)$ & $x^{13} + x^{12} + x^{11} + x^9 + t^2x^7 + t^2x^6 + x^4 + x^2 + x + 1$ & $[28, 15, 8]$  \\
	\hline
  $(30,4)$ & $x^7 + t^2x^6 + x^5 + tx^4 + tx^3 + x^2 + t^2x + 1$ & $[30, 23, 4]$  \\
	\hline

 $(30,4)$ & $x^{13} + tx^{12} + x^{11} + t^2x^{10} + tx^8 + x^7 + x^6 + tx^5 + t^2x^3 + x^2 +
    tx + 1$ & $[30, 17, 8]*$  \\
	\hline

\end{tabular}\caption*{$`*'$ represents the optimal codes}
\end{center}
\end{table}

 \section{Computational Results}
In this section, we provide some examples and tables for classical codes and DNA codes to support our study. Here we see that $(\theta,\delta)$-cyclic codes have better parameters than the cyclic and $\theta$-cyclic codes. Further, by taking the ring $\mathfrak{R}_l$, we can choose any number of copies to overcome the restriction on the length of the codes. Following are some examples of classical and DNA $(\sigma, \delta)$-codes over corresponding rings.
It's essential to note that the computations have been performed using Magma \cite{magma} and Sagemath \cite{sage}.
 \begin{example}
Let $\mathbb{F}_{2^4}=\mathbb{F}_{2}(t)$ be the field of order $16$. We define $\sigma \in Aut(\mathfrak{R}_2)$ as $\sigma(a+bv)=a^{2}+b^{2}v$ and let $\delta(a) =t(\sigma(a)-a)$ for some $a\in \mathfrak{R}_2$.
	Let $\mathcal{C}$ be a $(\theta, \delta)$-cyclic code of length $12$ over $\mathfrak{R}_2$ generated by $f(x)=\mu_1 f_1(x)+\mu_2 f_2(x)$, where $f_1(x)=x^4+t^{13}x^3 + t^{7}x^2 + t$ and $f_2(x)=x^3+t^{10}x^2 + t^{11}x + t^{14}$ are right divisors of $x^{12}-1$ in $\mathbb{F}_{16}[x;\theta,\delta]$. The factorization of $x^{12}-1$ in terms of $f_1(x)$ and $f_2(x)$ are given by
 \begin{align*}
     x^{12}-1=& (x^8 + t^{13}x^7 + t^2x^6 + t^7x^5 + t^2x^4 + t^{14}x^3 + t^5x^2 + t^6x + t^{11})(x^4 + t^{13}x^3 + t^7x^2 + t)\text{ and}\\
     x^{12}-1=& (x^9 + t^5x^8 + t^3x^7 + t^{10}x^5 + t^{14}x^4 + x^2 + x + 1)(x^3 + t^{10}x^2 + t^{11}x + t^{14})
 \end{align*}

	 Again, take  \[
	N=
	\left( \begin{array}{cc}
		1&t \\
		t&1
		\end{array}  \right)\in GL_{2}(\mathbb{F}_{2^4}),
	\] satisfying $NN^T=tI_2$. Then the Gray image $\Phi({\mathcal{C}})$ is a $(\theta, \delta)$-cyclic code with the parameters $[24,17,6]_{16}$ which is better than the code $[24,16,6]_{16}$\cite{shikha}.
\end{example}
 \begin{example}
Let $\mathbb{F}_{5^2}=\mathbb{F}_{5}(t)$ be the field of order $25$. We define $\sigma \in Aut(\mathfrak{R}_3)$ as $\sigma(a+bv+cv^2)=a^{5}+b^{5}v+c^5v^2$ and let $\delta(a) =t(\sigma(a)-a)$ for some $a\in \mathfrak{R}_3$.
	Let $\mathcal{C}$ be a $(\theta, \delta)$-cyclic code of length $15$ over $\mathfrak{R}_3$ generated by $f(x)=\mu_1 f_1(x)+\mu_2 f_2(x)+\mu_3 f_3(x)$, where $f_1(x)=x^3 + t^7x^2 + t^{22}x + t^9,~ f_2(x)=x + t^{11}$ and $f_3(x)=x + 4$ are right divisors of $x^{15}-1$ in $\mathbb{F}_{25}[x;\theta,\delta]$. The factorization of $x^{15}-1$ in terms of $f_1(x)$, $f_2(x)$ and $f_3(x)$ are given by
 \begin{align*}
     x^{15}-1=& (x^{12} + t^{19}x^{11} + 4x^{10} + t^{13}x^9 + t^{23}x^8 + t^8x^7 + t^{10}x^6 + t^7x^5 +
    t^{11}x^4 + t^{10}x^3 + t^{14}x^2 + t^2x\\& + t^5)(x^3 + t^7x^2 + t^{22}x + t^9),\\
    x^{15}-1=&(x^{14} + t^{23}x^{13} + t^{20}x^{12} + 4x^{11} + t^{21}x^{10} + t^9x^9 + tx^8 + t^{15}x^7 +
    t^{20}x^6 + t^{22}x^5 + t^{17}x^4 + t^{13}x^3\\& + 4x^2 + t^{15}x + t^5)(x + t^{11}) \text{ and}\\
    x^{15}-1=& (x^{14} + x^{13} + x^{12} + x^{11} + x^{10} + x^9 + x^8 + x^7 + x^6 + x^5 + x^4 + x^3 + x^2 + x + 1)(x + 4).
 \end{align*}
	 Again, take  \[
	N=
	\left( \begin{array}{ccc}
		t^{11}&4&t^{14} \\
		t^{17}&t^{17}&1 \\
            t^{10}& t^{17}&t^{23}
		\end{array}  \right)\in GL_{3}(\mathbb{F}_{25}),
	\] satisfying $NN^T=t^{20}I_3$. Then the Gray image $\Phi({\mathcal{C}})$ is an $(\theta, \delta)$-cyclic code with the parameters $[45, 40, 4]_{25}$.
\end{example}

 \begin{example}
Let $\mathbb{F}_{2^2}=\mathbb{F}_{2}(t)$ be the field of order $4$. We define $\theta \in Aut(\mathbb{F}_{2^2})$ as $\theta(a)=a^{2}$ and let $\delta(a) =t(\theta(a)-a)$ for some $a\in \mathbb{F}_{2^2}$.
	Let $\mathcal{C}$ be a $(\theta, \delta)$-cyclic code of length $30$ over $\mathbb{F}_{2^2}$ generated by $f(x)=x^{13} + tx^{12} + x^{11} + t^2x^{10} + tx^8 + x^7 + x^6 + tx^5 + t^2x^3 + x^2 +
    tx + 1$, where $f(x)$ is a right divisor of $x^{30}-1$ in $\mathbb{F}_{4}[x;\theta,\delta]$. The factorization of $x^{30}-1$ in terms of $f(x)$ is given by
 \begin{align*}
     x^{30}-1=& (x^{17} + t^2x^{16} + x^{14} + tx^{12} + t^2x^9 + tx^8 + t^2x^5 + x^3 + tx + 1)(x^{13} + tx^{12} + x^{11} + t^2x^{10}\\& + tx^8 + x^7 + x^6 + tx^5 + t^2x^3 + x^2 + tx + 1)
 \end{align*}

 Then the code $\mathcal{C}$ is a $(\theta, \delta)$-cyclic code with the parameters $[30,17,8]_{4}$ which is optimal according to \cite{Grassl}. Further, the polynomial $f(x)$ is a palindromic polynomial of odd degree. Hence, the code is reversible.
\end{example}

 \begin{example}
Let $\mathbb{F}_{2^2}=\mathbb{F}_{2}(t)$ be the field of order $4$. We define $\theta \in Aut(\mathbb{F}_{2^2})$ as $\theta(a)=a^{2}$ and let $\delta(a) =t(\theta(a)-a)$ for some $a\in \mathbb{F}_{2^2}$.
	Let $\mathcal{C}$ be a $(\theta, \delta)$-cyclic code of length $30$ over $\mathbb{F}_{2^2}$ generated by $f(x)=x^9 + t^2x^8 + t^2x^7 + x^6 + x^3 + t^2x^2 + t^2x + 1$, where $f(x)$ is a right divisor of $x^{12}-1$ in $\mathbb{F}_{4}[x;\theta,\delta]$. The factorization of $x^{12}-1$ in terms of $f(x)$ is given by
 \begin{align*}
     x^{12}-1=& (x^3 + tx^2 + t^2x + 1)(x^9 + t^2x^8 + t^2x^7 + x^6 + x^3 + t^2x^2 + t^2x + 1)
 \end{align*}

 Then the code $\mathcal{C}$ is a $(\theta, \delta)$-cyclic code with the parameters $[12,3,6]_{4}$. Further, the code $\mathcal{C}$ is a DNA code with codewords given in the Table \ref{DNA codewords}.
\end{example}

	\begin{table}
 \begin{center}
		\vspace{0.5cm}	
		\caption{DNA codewords from $(\sigma,\delta)$-cyclic codes}

  \label{DNA codewords}
		\begin{tabular}{|c|c|c|c|}
			
			\hline
		CCTAATCCTAAT & TTACCATTACCA & GGATTAGGATTA & TGAGCGTGAGCG \\

CTCACGCTCACG & ACGGCAACGGCA & CTTCAACTTCAA & CCCCCCCCCCCC \\

CGTGAGCGTGAG & GCCGAAGCCGAA & CGCTCACGCTCA & ATTAGGATTAGG \\

CCAGGACCAGGA & ACTCGCACTCGC & AGTTGAAGTTGA & CTGGTCCTGGTC \\

CAGATACAGATA & AACTTCAACTTC &CGGCTTCGGCTT &ATAGACATAGAC \\

TAATCCTAATCC &ATGTCTATGTCT &AAAAAAAAAAAA &AGGACCAGGACC \\

CATTACCATTAC &AGACATAGACAT &GTCTATGTCTAT &TTCGGCTTCGGC \\

GGCAACGGCAAC &AAGCCGAAGCCG &CTATGTCTATGT &TACAGATACAGA \\

GACCAGGACCAG &CGAAGCCGAAGC &GCCGTTGCCGT &GCGAGTGCGAGT \\

GTTGCCGTTGCC &CAACGGCAACGG &GTGCGAGTGCGA &GATACAGATACA \\

GGGGGGGGGGGG &GGTCCTGGTCCT &TAGGATTAGGAT &CCGTTGCCGTTG \\

TCTGTATCTGTA &GAAGTTGAAGTT &ACATAGACATAG &CACGCTCACGCT \\

TCCTGGTCCTGG &GAGTGCGAGTGC &AATGGTAATGGT &GCTTCGGCTTCG \\

TTGAAGTTGAAG &TTTTTTTTTTTT &TCGCACTCGCAC &TGTATCTGTATC \\

ACCATTACCATT &TGGTAATGGTAA &TATCTGTATCTG &ATCCTAATCCTA \\

GTAATGGTAATG &TCAACTTCAACT &AGCGTGAGCGTG &GCACTCGCACTC\\
			\hline
			
		\end{tabular}\label{tab2}
		 \end{center}
	\end{table}

\section{Conclusion}
In this paper, we have studied the $(\sigma,\delta)$-cyclic codes over $\mathfrak{R}_l=\mathbb{F}_{p^m}/\langle v^l-v\rangle$ where $l$ and $m$ are positive integers with $l> 1$. Further, we have explored the application of these codes in obtaining DNA codes over $\mathbb{F}_4$. First, we defined a Gray map to find correspondence between the elements of the ring and the pairs of nucleotides. Further, we presented the conditions for a code to be reversible and the DNA code. Additionally, we have constructed several DNA codes and classical codes of better parameters provided in Tables [\ref{overfield}, \ref{2copies}, \ref{Tbl3}, \ref{DNA codes}].

\section*{Acknowledgement}
	The authors are thankful for financial support to the Council of Scientific \& Industrial Research (CSIR), Govt. of India under File No. 09/1023(0027)/2019- EMR-1, the Prime Minister Research Fellows Scheme (PMRF-ID:
2702443), Ministry of Education, Govt. of India  and Indian Institute of Technology Patna for providing research facilities.
\section*{Declarations}
\textbf{Data Availability Statement}: The authors declare that [the/all other] data supporting the findings of this study are available within the article. Any clarification may be requested from the corresponding author.  \\
\textbf{Competing interests}: The authors declare that there is no conflict of interest regarding the publication of this manuscript.\\
\textbf{Use of AI tools declaration}
The authors declare that they have not used Artificial Intelligence (AI) tools to create this article.

\end{document}